\documentclass[letterpaper, 10 pt, conference]{ieeeconf}  

\IEEEoverridecommandlockouts                    
\usepackage{graphics} 
\usepackage{epsfig} 
\usepackage{mathptmx} 
\usepackage{times} 
\usepackage{amsmath} 
\usepackage{amssymb}  
\usepackage{subcaption}
\usepackage{units}
\usepackage{xcolor}
\usepackage{hyperref}
\newtheorem{definition}{Definition}

\newtheorem{proposition}{Proposition}
\newtheorem{assumption}{Assumption}
\newtheorem{theorem}{Theorem}

\IEEEoverridecommandlockouts

\usepackage[noadjust]{cite}

\usepackage[font=scriptsize]{caption}

\title{\LARGE \bf
Markovian Decentralized Ensemble Control for Demand Response}

\author{Guanze Peng, Robert Mieth, Deepjyoti Deka, and Yury Dvorkin}
\usepackage{placeins}
\newcommand{\subparagraph}{}
\usepackage{titlesec}

\titlespacing*{\section}{0pt}{5pt}{3pt}
\titlespacing*{\subsection}{0pt}{3pt}{0pt}
\titlespacing*{\subsubsection}{0pt}{3pt}{0pt}
\definecolor{frenchblue}{rgb}{0.0, 0.45, 0.73}

\begin{document}

\maketitle
\thispagestyle{empty}
\pagestyle{empty}

\begin{abstract}
With the advancement in smart grid and smart energy devices, demand response becomes one of the most economic and feasible solutions to ease the load stress of the power grids during peak hours. In this work, we propose a fully decentralized ensemble control framework with consensus for demand response (DR) events and compatible control methods based on random policies. We show that under the consensus that is tailored to DR, our proposed decentralized control method yields the same optimality as the centralized control method in both myopic and multistage settings. 
\end{abstract}

\section{Introduction}
{S}{upported} by the evolution of digital communication and control technology, electric power system operators deploy demand response (DR) programs to efficiently meet peak electricity demand and avoid network congestion. 
According to the US Federal Energy Regulatory Commission, potential peak demand reductions in the United States reached a total of approximately \unit[31]{GW} in 2019 \cite{lee2013assessment}. In the California Independent System Operator market, the total amount of demand response capacity registered grew from  0 to \unit[3.6]{GW} between June 2014 and June 2021 \cite{ciso2020}. 
DR programs engage distributed demand-side resources, e.g., controllable residential and commercial loads, by encouraging or directly controlling electricity customers to change their power consumption in exchange for an incentive.
To scale program efficiency, individual demand side resources (units) are often aggregated into ensembles operated by a DR aggregator or a utility/system operator.

\begin{figure}[t]
\centering
\begin{subfigure}[b]{0.47\textwidth}
  \includegraphics[width=\linewidth]{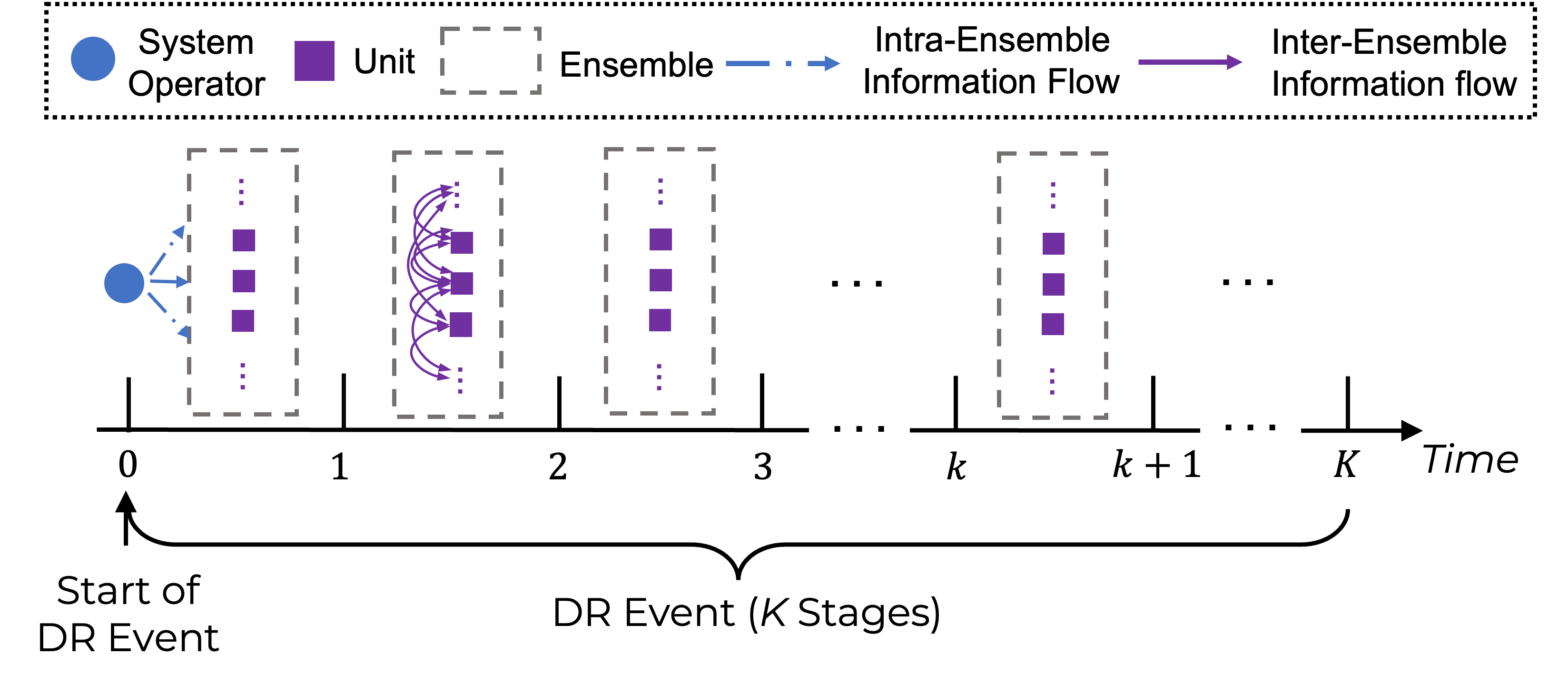}
		\caption{}\label{timeline}
\end{subfigure}
\begin{subfigure}[b]{0.45\textwidth}
  \includegraphics[width=\linewidth]{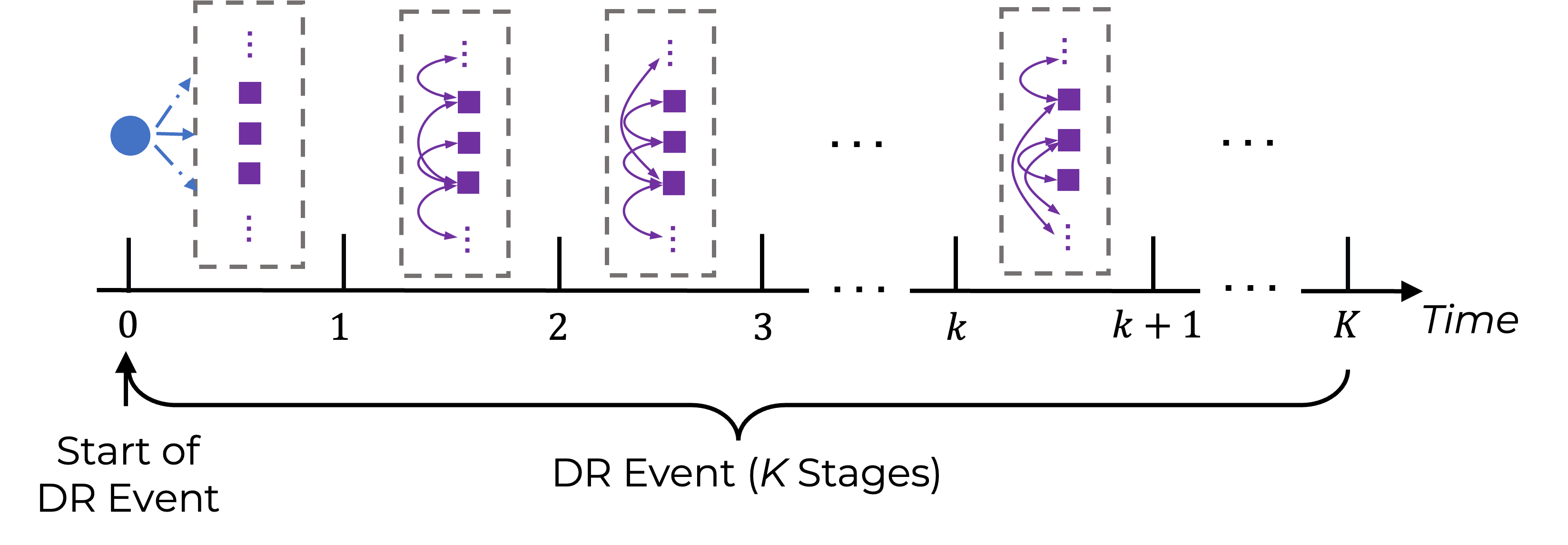}
		\caption{}\label{timeline2}
\end{subfigure}
\caption{{Decentralized ensemble control with different consensus schemes: in the $0$-\textit{th} stage of both schemes, the system operator broadcasts the message to the ensemble, informing them the start of DR event and the cost of each stage. (a) Global consensus scheme: units reach consensus in the first stage after being informed and do not communicate with each other in the following stages. (b) Local consensus scheme: units communicate with each other and reach consensus in each stage.}}\label{timelineall}
\end{figure}

Common DR architectures are centralized and set high communication and computation requirements, which requires the complete state information of all units in an ensemble, and all computation tasks to be concentrated and performed at the system operator end. \textcolor{black}{To improve computational feasibility, we propose a decentralized DR framework with consensus,  where units determine power consumption through negotiation among themselves, which aligns with the recent developments in smart power distribution system infrastructure. } 
DR decentralization is facilitated by the established communication protocol OpenADR 2.0 \cite{hassan2020hierarchical,openadr}, which is a protocol that standardizes the secure communication of dispatch signals (e.g., requests for load reduction) from the system operator to the individual units via a communication module.
However, as this communication module is also equipped with Virtual End Node, units can use OpenADR 2.0 to also communicate safely with each other. This enables a decentralized controller design, which can distribute the computation and communication load among all units.

Whenever DR is needed to support system operations, the system operator schedules a so-called DR event.
We call the time period between two valid consecutive OpenADR 2.0 dispatch commands the stage of a DR event. 
Since the dispatch commands are only valid for 15 minutes at most and DR events can last for hours, a DR event is a multistage process by definition. 
In the proposed framework, units are coupled through the consensus regarding their energy consumption policies. Particularly, we study two contrasting consensus schemes, \textit{global consensus} and \textit{local consensus}, both of which are illustrated in Fig.~\ref{timelineall}. 
Under the former, communication and consensus-forming between units is limited to the first stage. 
The latter assumes that units can communicate and reach consensus at every stage.
We show that the total communication throughput for the local consensus scheme is significantly reduced compared to global consensus, which facilitates practical implementation.
We summarize our main contributions as follows:
\begin{itemize}
	\item We propose a decentralized ensemble control framework with consensus. We leverage the Bayesian prior stage concept and propose a framework with prior consensus such that consensus is achieved among all units;
	\item We show that the designed decentralized controller is equivalent to a corresponding centralized one;
	\item We show that the global and local multistage consensuses yield the same energy consumption policy;
	\item We introduce a distributed gradient based algorithm that allows units to reach consensus.
\end{itemize}

We build on existing work on randomized DR control, e.g.,  \cite{chen2014individual, meyn2015ancillary, hao2012demand}, and leverage Linearly Solvable Markov decision processes (LS-MDP) \cite{todorov2006linearly} to model the behavior of a DR ensemble. 
LS-MDP have been shown to be particularly suitable for thermostatically controlled loads 
\cite{chertkov2018ensemble}, which play a central role for real-world DR programs as the thermal inertia of cooling and heating systems allows temporally adjusting power consumption without compromising primary system functions and causing large discomfort to electricity customers.   
Related work in \cite{hassan2020stochastic,pop2019markov} derived and applied optimal control policies for thermostatically controlled load (TCL) ensembles, modeled as LS-MDPs, to support power system operations, while minimizing the difference from individual default energy consumption patterns and thus the overall discomfort. 
\textcolor{black}{Comparing with existing decentralized control method \cite{hassan2020stochastic,pop2019markov, koch2011modeling}, our proposed model provides an elegant analytical solution which is easy to implement and can be easily extended to other demand response scenarios such as heating systems. }The decentralized consensus-based approach proposed in this paper is closely related to the decentralized LS-MDP \cite{wan2021distributed}, which studied a network of Bayesian agents in continuous Markov Processes (MP) without consensus. The introduced algorithm resembles distributed gradient algorithms over (random) networks \cite{lobel2010distributed,zhang2014online}.


\section{Preliminaries}

Consider an ensemble of $N$ units given by set $\mathcal{N}=\{1,2,...,N\}$. {\color{black}We model the operation of each unit using a shared set of states $\mathcal{S}=\{1,2,...,S\}$ reflecting its power consumption.} 
We assume in this paper:
 \begin{assumption}
\textit{Every unit can transition from any state in $\mathcal{S}$ to any state in $\mathcal{S}$. Moreover, the $n$-th uncontrolled unit at state $i$ transitions to state $j$ in the next stage with probability $\bar{p}^{(n)}_{ij}\in(0,1)$ such that $\sum_{j\in\mathcal{S}}\bar{p}^{(n)}_{ij}=1$, $i\in\mathcal{S}$, $n\in\mathcal{N}$.}
\end{assumption}

As per Assumption~1 each unit has a default randomized policy for each state, which captures the default energy consumption behavior of that unit prior to the DR event. Any deviation from this default behavior causes ``discomfort'' to the DR participant, which we measure by the difference between the controlled policy and the default policy.
We denote the cost of per unit discomfort of the $n$-th unit at state $i$ by $\gamma_i^{(n)}>0$, $i\in\mathcal{S}$, $n\in\mathcal{N}$. 
{\color{black}Let $q_{i}$, $i\in\mathcal{S}$ be the dollar cost of operation at state~$i$, i.e., power consumption times cost per unit of power.}
The cost function of the $n$-th agent in state $i$, $c^{(n)}(\cdot\,;\,\cdot):\Delta_S\times\mathcal{S}\rightarrow\mathbb{R}$, is given by
\begin{equation}\label{onestage}
	c^{(n)}(\mathbf{p}_i;i)=\sum_{j\in\mathcal{S}}p_{ij}q_{j}+\gamma^{(n)}_i\sum_{j\in\mathcal{S}} p_{ij}\ln\frac{p_{ij}}{\bar{p}^{(n)}_{ij}},\quad n\in\mathcal{N},\,i\in\mathcal{S},
\end{equation}
where $\mathbf{p}_i:=[p_{ij}]_{j\in\mathcal{S}}\in\Delta_\mathcal{S}$ is the controlled randomized policy\textcolor{black}{, and $\Delta_\mathcal{S}:=\{[\eta_{i}]_{i\in\mathcal{S}}:0\leq \eta_i\leq 1, i\in\mathcal{S},\sum_{i\in\mathcal{S}}\eta_i=1\}$ is the simplex over the finite set $\mathcal{S}$.} For the sake of readability, we omit stage index until Section IV.
The first term in \eqref{onestage} quantifies the expected physical energy cost and the second term quantifies the expected discomfort cost. 
We adopt the popular and effective Kullback–Leibler divergence \cite{kullback1951information} to measure the expected policy discomfort level.


\section{Decentralized Control with Consensus}
\textcolor{black}{ As a foundation to deriving the multistage consensus process, we first consider all units to be myopic, i.e., they only plan their consumption one stage ahead.}
We call the collection of states of units \textit{joint state}, denoted by $(i_n,\ n\in\mathcal{N})$, where $i_n\in\mathcal{S}$ is the state of the $n$-th user. \textcolor{black}{Let $i_n$ denote the state of unit $n$ and $\mathcal{N}_i=\{n\in\mathcal{N}:i_n=i\}$ denote the index set of units whose state is $i$ with $|\mathcal{{N}}_i|=N_i$.} Given the joint state, the one-stage optimization with consensus can be written as
\begin{equation}\label{CDO1}
\begin{aligned}
	\min_{\mathbf{p}^{(n)}_{i_n},\ n\in\mathcal{N}}&\quad \frac{1}{N}\sum_{n\in\mathcal{N}}{c}^{(n)}(\mathbf{p}^{(n)}_{i_n};{i_n})\\
	{s.t.}	&\quad \mathbf{p}^{(m)}_{i_m}=\mathbf{p}^{(n)}_{i_n},\quad \textcolor{black}{\forall}\, i_m=i_n\in\mathcal{S},\,\textcolor{black}{\forall}\, m,n\in\mathcal{N},\\
	&\quad 
	\mathbf{p}^{(n)}_{i_n}\in\Delta_\mathcal{S},\quad\ \textcolor{black}{\forall}\, n\in\mathcal{N},
\end{aligned}
\end{equation}
where the objective function is the average cost of all units. The first constraint indicates that two that are in the same state reach consensus on using the same randomized policy.
{\color{black}As a result, consensus between all units ensures least-cost and non-discriminatory ensemble operations,  even if cost functions $c^{(n)}$ are different for individual units.}
The solution generalizes LS-MDPs of homogeneous units  \cite{chertkov2018ensemble}.
\begin{theorem}\label{zero}
\textit{The optimal solution of \eqref{CDO1} is given by
\begin{equation}\label{optmat}
	{p}^*_{ij}=\frac{e^{({\mu}_{ij}-{q}_{ij})/{\gamma}_i}}{\sum_{s\in\mathcal{S}}e^{({\mu}_{is}-{q}_{is})/{\gamma}_i}},\quad i\,,j\in\mathcal{S},
\end{equation}
where
$
	{\gamma}_i=\frac{1}{N}\sum_{n\in\mathcal{N}_i}\gamma^{(n)}_{i},\,{\mu}_{ij}=\frac{1}{N}\sum_{n\in\mathcal{N}_i}\gamma^{(n)}_{i}\ln \bar{p}^{(n)}_{ij},$ and \textcolor{black}{$\ \theta_{i}=\frac{1}{N}\sum_{n\in\mathcal{N}_i}q_i$}.
Moreover, the optimal value of \eqref{CDO1} is given by:
\begin{equation}\label{optinvolv}
	c^*(i_n,n\in\mathcal{N}_i)=-\sum_{i\in\mathcal{S}}{\gamma}_i\ln\left(\sum_{j\in\mathcal{S}}e^{({\mu}_{ij}-\theta_{i})/{\gamma}_{ij}}\right).
\end{equation}
}
\end{theorem}

\begin{proof}
\textcolor{black}{ Assuming that all units can reach consensus via information exchange, as dictated in \eqref{CDO1}, the optimal solution can be found by the following reformulation of \eqref{CDO1}:
}
	\begin{equation}\label{postopt}
\begin{aligned}
	\min_{\mathbf{p}_{i},i\in\mathcal{S}}&\quad \frac{1}{N}\sum_{i\in\mathcal{S}}\sum_{n\in\mathcal{N}_i}{c}^{(n)}(\mathbf{p}_{i};{i})\\
	\text{s.t.}	&\quad 	\mathbf{p}_{i}\in\Delta_\mathcal{S},\quad\ \ \textcolor{black}{\forall}\,i\in\mathcal{S},
\end{aligned}
\end{equation}
	for which we derive the following Lagrangian:
	\begin{equation*}
		\begin{aligned}
			&\mathcal{L}(\{\mathbf{p}_i\}_{i\in\mathcal{S}},,\{\lambda_i\}_{i\in\mathcal{S}};\{i_n\}_{n\in\mathcal{N}})\\
			&\qquad\qquad\qquad\qquad=\sum_{i\in\mathcal{S}} \left(\frac{1}{N}\sum_{n\in\mathcal{N}}c^{(n)}(\mathbf{p}_i)+\lambda_i\mathbf{1}^{\rm T}_S\mathbf{p}_i\right),
		\end{aligned}
	\end{equation*}
	where $\lambda_i$ is the Lagrangian multiplier for state $i\in\mathcal{S}$, and $\mathbf{1}_S=[1,1,...,1]^{\rm T}\in\mathbb{R}^S$. By the first-order condition, we obtain
	\begin{equation}\label{mattmp}
		p_{ij}^*=e^{({\mu}_{ij}-\lambda_i-\theta_i)/{\gamma}_i-1}.
	\end{equation}
	By \eqref{mattmp} and the constraint in \eqref{postopt}, we obtain \eqref{optmat}, which yields the optimal value \eqref{optinvolv}. 
\end{proof}

When the ensemble is large, the joint state $(i_n,\ n\in\mathcal{N})$ results in a large state space, which is difficult to accommodate in the multistage case. Moreover, the analysis of \eqref{CDO1} is complicated by the non-linearity of  the optimal solution and the optimal value in the joint state $(i_n$, $n\in\mathcal{N})$. 
\textcolor{black}{Also, units are forced to reach consensus only with units in the same state:} A closer inspection of Theorem~\ref{zero} prompts that \eqref{CDO1} can be decomposed into $S$ \textit{sub-consensus} problems. For each $i\in\mathcal{S}$, let the $i$-th sub-consensus problem of \eqref{CDO1} be
\begin{equation}\label{CDO2}
\begin{aligned}
	\min_{\mathbf{p}^{(n)}_{i},n\in\mathcal{N}_i}&\quad \frac{1}{N}\sum_{n\in\mathcal{N}_i}{c}^{(n)}(\mathbf{p}^{(n)}_{i};{i})\\
	\text{s.t.}	&\quad \mathbf{p}^{(m)}_{i}=\mathbf{p}^{(n)}_{i},\quad\ \textcolor{black}{\forall}\,m,\,n\in\mathcal{N}_i,\\
	&\quad 
	\mathbf{p}_{i}^{(n)}\in\Delta_\mathcal{S},\quad\ \ \textcolor{black}{\forall}\,n\in\mathcal{N}_i.
\end{aligned}
\end{equation}
The first constraint implies that each unit reaches consensus with other units at the same state, rather than with all units in the ensemble. \textcolor{black}{A trivial consensus scheme where units implement the same randomized policy for every state, }
\begin{equation}\label{trivial}
\begin{aligned}
	\min_{\mathbf{p}^{(n)},n\in\mathcal{N}}&\quad \frac{1}{N}\sum_{n\in\mathcal{N}}{c}^{(n)}(\mathbf{p}^{(n)};{i_n})\\
	\text{s.t.}	&\quad \mathbf{p}^{(m)}=\mathbf{p}^{(n)},\quad\ \textcolor{black}{\forall}\,m,n\in\mathcal{N},\\
	&\quad 
	\mathbf{p}^{(n)}\in\Delta_\mathcal{S},\quad\ \ \textcolor{black}{\forall}\,n\in\mathcal{N},
\end{aligned}
\end{equation}
\textcolor{black}{can assure that all units reach consensus with each other, yet yielding a greater or equal cost compared to \eqref{CDO1}. Therefore, both \eqref{CDO1} and \eqref{trivial} have significant drawbacks.}

To balance the caveats of \eqref{CDO1} and the sub-optimality of  \eqref{trivial}, we propose a fully decentralized framework using Bayesian approach, which achieves consensus across the entire ensemble. To elaborate our idea, we introduce two important concepts, prior consensus and posterior consensus.

\subsection{Prior and Posterior Consensus}

Consider a centralized randomized ensemble control scenario, where a system operator controls an ensemble with the optimal solution of \eqref{CDO1}. \textcolor{black}{Let $x_i:={N}_i/N$ be the percentage of units at state $i$ in the ensemble}, and we call $\mathbf{x}=[x_{i}]_{i\in\mathcal{S}}\in\Delta_\mathcal{S}$ the \textit{ensemble state}. Units are assumed to implement their randomized policy independently and, thus, each unit is subject to the randomness of its randomized policy independently. As a result, the decentralized control in \eqref{CDO1} has the same ensemble state transition  as its centralized counterpart \eqref{postopt}.

To obtain the optimization of the form \eqref{CDO1}, the system operator needs to know the joint state $(i_n$, $n\in\mathcal{N})$. 
We call a system operator that knows the full joint state \textit{informed} and a system operator that only knows the ensemble state $\mathbf{x}$ \textit{uninformed}.
For an uninformed Bayesian system operator, prior belief is given by ensemble state $\mathbf{x}$ yielding
\begin{equation}\label{PriorBayes}
\begin{aligned}
	\textcolor{black}{\min_{\mathbf{p}_i\in\Delta_\mathcal{S}, i\in\mathcal{S}}\quad \frac{1}{N}\sum_{n\in\mathcal{N}}\mathbb{E}_{i\sim\mathbf{x}}\left[{c}^{(n)}(\mathbf{p}_i;{i})\right].}	\end{aligned}
\end{equation}
If we regard state $i_n$, $n\in\mathcal{S}$, as the signal in the Bayesian framework, which is drawn from identically independent distributions determined by ensemble state $\mathbf{x}$, then \eqref{PriorBayes} is a Bayesian optimization problem in the prior stage \cite{greenberg2012introduction}. 

For the informed system operator, since all signals are known, the centralized scheme for \eqref{CDO1} yields \eqref{postopt}, which is a Bayesian optimization at the posterior stage with beliefs
\begin{equation*}
	\mathbb{P}[s_n=i|s_n=i_n]=\begin{cases}	
1 &{i=i_n},\\
0 &\text{otherwise},
\end{cases}
\end{equation*}
where $s_n$ is a random variable representing the state of the \mbox{$n$-th} unit. Indeed, \eqref{PriorBayes} and \eqref{postopt} are coupled by Bayes' rule as:
\begin{equation*}
	\mathbb{P}[s_n=i|s_n=i_n]=\frac{\mathbb{P}[s_n=i_n|s_n=i]\mathbb{P}[s_n=i]}{\sum_{j\in\mathcal{S}}\mathbb{P}[s_n=i_n|s_n=j]\mathbb{P}[s_n=j]},
\end{equation*}  
where $\mathbb{P}[s_n=i]=x_i$, $i\in\mathcal{S}$. Moreover, since \eqref{PriorBayes} equals
$$\min_{\mathbf{P}\in\mathcal{P}}\quad \frac{1}{N}\sum_{n\in\mathcal{N},i\in\mathcal{S}}x_i\,{c}^{(n)}(\mathbf{p}_i;i),
$$
where $\mathbf{P}:=[\mathbf{p}^{\rm T}_i]_{i\in\mathcal{S}}$ and \eqref{PriorBayes} can be viewed as a linearization of \eqref{postopt} using Bayesian beliefs.

\subsection{Decentralized Ensemble Control with Prior Consensus}

Based on the discussion in  Section III-A, we define
\begin{definition}
\textit{In decentralized consensus problems, units reach prior consensus, if the randomized policies of units converge to the optimal solution of \eqref{PriorBayes}, and units reach posterior consensus, if the randomized policies of units converge to the optimal solution of \eqref{CDO1}.} \end{definition}
 
In a decentralized framework with prior consensus, we have the following optimization:
\begin{equation}\label{priorop}
\begin{aligned}
	\min_{\mathbf{p}^{(n)}_i,i\in\mathcal{S},n\in\mathcal{N}}\quad&\frac{1}{N}\sum_{n\in\mathcal{N}}\mathbb{E}_{i\sim\mathbf{x}}\left[{c}^{(n)}(\mathbf{p}^{(n)}_i;{i})\right]\\
	\text{s.t.}\quad&\mathbf{p}^{(n)}_i=\mathbf{p}^{(m)}_i,\quad \textcolor{black}{\forall\,}i\in\mathcal{S},\,\textcolor{black}{\forall\,} m,n\in\mathcal{N},\\
	&\mathbf{p}_i^{(n)}\in\Delta_\mathcal{S},\quad \ \textcolor{black}{\forall\,}i\in\mathcal{S},\, \textcolor{black}{\forall\,}n\in\mathcal{N}.
	\end{aligned}
\end{equation}
By using $\mathbb{E}_{i\sim\mathbf{x}}[\,\cdot\,]$, each unit disregards its own state information and only relies on the prior distribution of the whole ensemble state to reach prior consensus. 

\begin{theorem}\label{Priortheorem}
\textit{The optimal solution of \eqref{priorop}, $\mathbf{P}^*=[p_{ij}^*]_{i,j\in\mathcal{S}}$, exists and is given by
\begin{equation*}
	{p}^*_{ij}=\frac{e^{(\bar{\mu}_{ij}-q_j)/\bar{\gamma}_i}}{\sum_{s\in\mathcal{S}}e^{(\bar{\mu}_i-q_s)/\bar{\gamma}_i}},\quad j\in\mathcal{S},
\end{equation*}
where $\bar{\gamma}_i=\frac{1}{N}\sum_{n\in\mathcal{N}}\gamma^{(n)}_i$ and $\bar{\mu}_{ij}=\frac{1}{N}\sum_{n\in\mathcal{N}}\gamma^{(n)}_i\ln \bar{p}^{(n)}_{ij}$. Moreover, the optimal value of \eqref{CDO1} is
\begin{equation*}
	c^*(\mathbf{x})=-\sum_{i\in\mathcal{S}}x_i\bar{\gamma}_i\ln\left(\sum_{j\in\mathcal{S}}e^{(\bar{\mu}_{ij}-q_j)/\bar{\gamma}}_i\right).
\end{equation*}}
\end{theorem}

\begin{proof}
The proof is similar to that of Theorem~\ref{zero} with the following Lagrangian:
	\begin{equation*}
		\mathcal{L}(\{\mathbf{p}_i\}_{i\in\mathcal{S}},\{\lambda_i\}_{i\in\mathcal{S}};\mathbf{x})=\sum_{i\in\mathcal{S}}\left(x_i\frac{1}{N}\sum_{n\in\mathcal{N}}c^{(n)}(\mathbf{p}_i)+\lambda_i\mathbf{1}^{\rm T}_S\mathbf{p}_i\right).
	\end{equation*}
	By the first-order condition, we obtain
	\begin{equation*}
		p_{ij}^*=e^{(\bar{\mu}_{ij}-\lambda_i/x_i-\bar{q}_{ij})/\bar{\gamma}_i-1}.
	\end{equation*}
	The rest of the proof is analogous to Theorem 1.
\end{proof}

As shown in Theorem 2, the optimal solution is more interpretable and scalable than that in Theorem~\ref{zero}, because the optimal value of \eqref{priorop} is linear in the prior belief/ensemble state and the optimal solution is independent of the prior belief/ensemble state. Furthermore, there does not exist sub-consensus, as the $\bar{\gamma}_i$ and $\bar{\mu}_{ij}$ are averaged over all units for each $i,j\in\mathcal{S}$ in contrast to \eqref{CDO1}. Compared to \eqref{trivial}, \eqref{priorop} results in a lower cost as it allows units to execute state-dependent randomized policies.
{\color{black}Note that auxiliary terms $\bar{\gamma}_i$ and $\bar{\mu}_{ij}$ do not need to be communicated among the units as we show in \eqref{valuefns} below.}

\section{Multistage Consensus}\label{}

This section focuses on the multistage decentralized optimization problem with prior consensus. We first discuss the DR problem with \textit{multistage global consensus}. 

\subsection{Global Consensus}
In the global consensus scheme, units face an open-loop optimization problem, where they can reach consensus only once before choosing a randomized policy. For every $n\in\mathcal{N}$ and $\mathbf{x}=[x_i]_{i\in\mathcal{S}}\in\Delta_{\mathcal{S}},$ define ${c}^{(n)}_\ell:\mathcal{P}\times\Delta_\mathcal{S}\rightarrow\mathbb{R}$ as
\begin{equation*}
	{c}^{(n)}_\ell(\mathbf{P},\mathbf{x})=\sum_{i\in\mathcal{S}}x_i\left(q_{i,\ell}+\gamma^{(n)}_{i,\ell}\sum_{j\in\mathcal{S}} p_{ij}\ln\frac{p_{ij}}{\bar{p}^{(n)}_{ij,\ell}}\right),\ 1\leq \ell<L,
\end{equation*}
which is an average cost function of the $n$-th unit evaluated over $\mathbf{x}$ as prior distribution in stage $\ell$. \textcolor{black}{Here, $\mathcal{P}:=\Delta_\mathcal{S}\times \dots\times\Delta_\mathcal{S}\subseteq\mathbb{R}^{S\times S}$. } Let ${c}^{(n)}_L(\mathbf{x})=\sum_{i\in\mathcal{S}}x_i\,\cdot\,q_{i,L}$ be the terminal cost for units. The consensus optimization becomes:
 \begin{equation}\label{global}
	\begin{aligned}
		 &\min_{\{\mathbf{P}^{(n)}_{\ell}\}_{\ell=1}^{L-1}}\ \  \frac{1}{N}\sum_{n\in\mathcal{N}}\mathbb{E}_1\left[\sum_{\ell=1}^{L-1}{c}^{(n)}_\ell(\mathbf{P}^{(n)}_\ell,\mathbf{x}^{(n)}_\ell)+{c}^{(n)}_L(\mathbf{x}^{(n)}_{L})\right],\\
		 &\ \ \text{s.t.}\ \  \mathbf{x}^{(n)}_{\ell+1}=\left(\mathbf{P}^{(n)}_{\ell}\right)^{\rm T}\mathbf{x}^{(n)}_{\ell}+\mathbf{Q}_{t+1}^{(n)},\  \textcolor{black}{\forall\,}n\in\mathcal{N},\textcolor{black}{\forall\,}1\leq\ell< L,\\
		  &\ \ \quad\ \ \, \mathbf{P}^{(m)}_\ell=\mathbf{P}^{(n)}_\ell,\quad\quad\quad\quad\quad\ \ \textcolor{black}{\forall}\,m,n\in\mathcal{N},\textcolor{black}{\forall}\,1\leq\ell< L,\\
		 &\ \ \quad\ \ \, \mathbf{P}^{(n)}_\ell\in\mathcal{P},\quad\quad\quad\quad\quad\quad\quad\ \ \textcolor{black}{\forall}\,n\in\mathcal{N},\textcolor{black}{\forall}\,1\leq\ell< L,
		 \end{aligned}	
\end{equation}
where the expectation $\mathbb{E}_{1}[\cdot]$ is taken over future ensemble states $\mathbf{x}_\ell$, $1<\ell\leq L$, with prior belief $\mathbf{x}_1$.
 The first constraint stems from the implementation of randomized policies, \textcolor{black}{where $\mathbf{Q}_t^{(n)}\sim(\mathbf{0},\text{diag}[\frac{1}{N}\sum_{i\in\mathcal{S}}x_i^2(p_{ij}-p_{ij}^2)]_{j\in\mathcal{S}})$ is the induced Gaussian noise}\footnote{\textcolor{black}{The ensemble state is approximated by Gaussian random variables based on the central limit theorem \cite{jacod2004probability}.}}. The second constraint is the consensus constraint requiring all units to reach consensus on the policy matrix. The last constraint guarantees that  decision variable $\mathbf{P}$ is a valid policy matrix.

To find the optimal solution of \eqref{global}, consider an uninformed system operator, controlling the ensemble in the open-loop fashion as
\begin{equation}\label{globalcen}
	\begin{aligned}
		 \min_{\{\mathbf{P}_{\ell}\}_{\ell=1}^{L-1}}\quad &\frac{1}{N}\sum_{n\in\mathcal{N}}\mathbb{E}_1\left[\sum_{\ell=1}^{L-1}{c}^{(n)}_\ell(\mathbf{P}_\ell,\mathbf{x}_\ell)+{c}^{(n)}_L(\mathbf{x}_{L})\right],\\
		 \text{s.t.}\quad & \mathbf{x}_{\ell+1}=\mathbf{P}_{\ell}^{\rm T}\mathbf{x}_{\ell}+\mathbf{Q}_{\ell+1},\quad\ \, \textcolor{black}{\forall}\,1\leq\ell< L,\\
		 & \mathbf{P}_\ell\in\mathcal{P},\quad\quad\quad\quad\quad\ \ \ \ \textcolor{black}{\forall}\,1\leq\ell< L,
		 \end{aligned}	
\end{equation}
To proceed, let the system operator be equipped with an increasing sequence of \textcolor{black}{$\sigma$}-algebras \textcolor{black}{\cite{jacod2004probability}}:
\begin{equation*}
\mathcal{I}_{\ell}=\begin{cases}
	\sigma(\{\mathbf{x}_{1}\}) &{\ell = 1,}\\
	\sigma(\{\mathbf{x}_{t}\}_{t=1}^{\ell},\{\mathbf{P}_{t}\}_{t=1}^{\ell-1}) &{2\leq \ell< L,}
\end{cases}	
\end{equation*}
such that \textcolor{black}{$\mathcal{I}_{\ell_1}\subseteq\mathcal{I}_{\ell_2}$ for all $1\leq \ell_1\leq \ell_2< L$. Then,} we can rewrite the objective function in \eqref{globalcen} as
\begin{equation*}
	\mathbb{E}_1\left[\sum_{\ell=1}^{L-1}\mathbb{E}_\ell\left[\frac{1}{N}\sum_{n\in\mathcal{N}}{c}^{(n)}_\ell(\mathbf{P}_\ell,\mathbf{x}_\ell)+\frac{1}{N}\sum_{n\in\mathcal{N}}{c}^{(n)}_L(\mathbf{x}_L)\right]\right],
\end{equation*}
where $\mathbb{E}_\ell[\,\cdot\,]:=\mathbb{E}[\,\cdot\,|\mathcal{I}_\ell]$. Moreover, define $v_{\ell}(\,\cdot\,):\Delta_\mathcal{S}\rightarrow\mathbb{R}$ as
\begin{equation}\label{equi}
	v_{\ell}(\mathbf{x}_\ell)=\min_{\{\mathbf{P}_{t}\}_{t=\ell}^{L-1}}\ \mathbb{E}_\ell\left[\frac{1}{N}\sum_{t=\ell}^{L-1}{c}^{(n)}_{t}(\mathbf{P}_{t},\mathbf{x}_{t})+{c}^{(n)}_L(\mathbf{x}_{L})\right].
\end{equation}
Here, $v_{\ell}(\,\cdot\,)$ is the value function that satisfies the following Bellman equation: for $1\leq\ell<L$, ${v}_{\ell}(\mathbf{x}_\ell)=\min_{\mathbf{P}_\ell\in\mathcal{P}}{V}_\ell(\mathbf{P}_\ell,\mathbf{x}_\ell)$, where ${V}_\ell(\mathbf{P}_\ell,\mathbf{x}_\ell)=\mathbb{E}_\ell\left[\frac{1}{N}\sum_{n\in\mathcal{N}}{c}^{(n)}_\ell(\mathbf{P}_\ell,\mathbf{x}_\ell)+{v}_{\ell+1}\left(\mathbf{x}_{\ell+1}\right)\right].$ The optimal solution of \eqref{globalcen} is then:
\begin{theorem}\label{dytheorem}
\textit{The optimal solution of \eqref{globalcen}, $\mathbf{P}_{\ell}^*=[p_{ij,\ell}^*]_{i,j\in\mathcal{S}}$, $1\leq\ell<L$, exists, and is given by
\begin{equation*}
	{p}^*_{ij,\ell}=\frac{e^{(\bar{\mu}_{ij,\ell}-v_{j,\ell+1})/\bar{\gamma}_{i,\ell}}}{\sum_{s\in\mathcal{S}}e^{(\bar{\mu}_{i,\ell}-v_{s,\ell+1})/\bar{\gamma}_{i,\ell}}},\quad j\in\mathcal{S},
\end{equation*}
where $\bar{\gamma}_{\ell}=\frac{1}{N}\sum_{n\in\mathcal{N}}\gamma^{(n)}_{i,\ell},$ and $\bar{\mu}_{ij,\ell}=\frac{1}{N}\sum_{n\in\mathcal{N}}\gamma^{(n)}_{\ell}\ln \bar{p}^{(n)}_{ij,\ell}$. Moreover, the value function of \eqref{globalcen} in stage $\ell$ is given by
\begin{equation}\label{valuemulti}
	v_\ell(\mathbf{x})=\sum_{i\in\mathcal{S}}x_iq_{i,\ell}-\sum_{i\in\mathcal{S}}x_i\bar{\gamma}_{i,\ell}\ln\left(\sum_{j\in\mathcal{S}}e^{(\bar{\mu}_{ij,\ell}-v_{j,\ell+1})/\bar{\gamma}_{i,\ell}}\right),
	\end{equation}
	with $v_L(\mathbf{x})=c_L(\mathbf{x}).$}
\end{theorem}

\begin{proof}
	In stage $L-1$, all units can be considered as myopic. A direct application of Theorem~\ref{Priortheorem} gives
	\allowdisplaybreaks
	\begin{equation*}
	v_{L-1}(\mathbf{x})=\sum_{i\in\mathcal{S}}x_iq_{i,L-1}-\sum_{i\in\mathcal{S}}x_i\bar{\gamma}_{i,L-1}\ln\sum_{j\in\mathcal{S}}e^{(\bar{\mu}_{ij,L-1}-q_{j,L})/\bar{\gamma}_{i,L-1}}.
	\end{equation*}
	Suppose that \eqref{valuemulti} holds for $\ell+1$. In stage $\ell$, 
	\begin{equation*}
	\begin{aligned}
	v_{\ell}(\mathbf{x})=\sum_{i\in\mathcal{S}}x_iq_{i,\ell} +\min_{\mathbf{P}\in\mathcal{P}}\sum_{i,j\in\mathcal{S}}x_ip_{ij}\left(v_{j,\ell+1}+	\bar{\gamma}_{i,\ell}\ln p_{ij}-\bar{{\mu}}_{ij,\ell}\right),
	\end{aligned}
	\end{equation*}
where $v_{i,\ell+1} =q_{j,\ell+1}- \bar{\gamma}_{i,\ell+1}\ln\left(\sum_{j\in\mathcal{S}}e^{(\bar{\mu}_{ij,\ell+1}-v_{j,\ell+2})/\bar{\gamma}_{i,\ell+1}}\right).$
	Hence, we can use Theorem~\ref{Priortheorem} again by replacing $q_j$ with $v_{j,\ell+1}$, and by backwards induction, the theorem follows. 
\end{proof}

The theorem above provides a useful insight. The optimal solution of \eqref{globalcen} is independent of the ensemble state. 
In fact, it follows from Theorem~\ref{Priortheorem} that the optimal policy $\mathbf{P}^*_\ell$ does not require the knowledge of the estimate $\mathbf{x}_\ell$. Only the computation of $\mathbf{x}_{\ell+1} = \mathbf{P}^{*\rm T}_\ell\mathbf{x}_\ell$ requires it. In other words, in the backward-forward algorithm, $\mathbf{x}_\ell$ is used only in the forward part of the multistage algorithm. This is beneficial as then value function $v_{\ell}(\mathbf{x}_\ell)$ is linear in $\mathbf{x}_\ell$ and can generate a tractable policy. As a matter of fact, we can distribute the consensus reaching process across all the stages as indicated in Fig.~\ref{timeline2},  which not only reduces the computation load of each unit, but also lowers the communication overhead.

\subsection{Local Consensus}

Suppose that units only reach consensus in each stage, i.e., reach \textit{multistage local consensus} as shown in Fig.\ref{timeline2}. 
Using ${v}^{(n)}_{L}\left(\mathbf{x}_{L}\right)=c^{(n)}_L\left(\mathbf{x}_{L}\right)$, for $1\leq\ell<L$, the consensus problem in stage $\ell$, $1\leq\ell<L$ is:
\allowdisplaybreaks
\begin{equation}\label{local}
	\begin{aligned}
		 \min_{\mathbf{P}_\ell^{(n)},n\in\mathcal{N}}\quad &\textcolor{black}{\frac{1}{N}\sum_{n\in\mathcal{N}}\mathbb{E}_\ell\left[{c}^{(n)}_{\ell}(\mathbf{P}^{(n)}_\ell,\mathbf{x}_{\ell}^{(n)})+{v}^{(n)}_{\ell+1}\left(\mathbf{x}^{(n)}_{\ell+1}\right)\right]}\\
		  \text{s.t.}\quad & \mathbf{x}^{(n)}_{\ell+1}=\left(\mathbf{P}^{(n)}_{\ell}\right)^{\rm T}\mathbf{x}^{(n)}_{\ell}+\mathbf{Q}_{\ell+1}^{(n)},\quad\textcolor{black}{\forall}\, n\in\mathcal{N},\\
		  & \mathbf{P}^{(m)}_\ell=\mathbf{P}^{(n)}_\ell,\quad\quad\quad\quad\quad\quad\quad\ \, \textcolor{black}{\forall}\, m,n\in\mathcal{N},\\
		 & \mathbf{P}^{(n)}_\ell\in\mathcal{P},\quad\quad\quad\quad\quad\quad\quad\quad\ \textcolor{black}{\forall}\,n\in\mathcal{N},
		 \end{aligned}	
\end{equation}%
\allowdisplaybreaks[0]%
where
${v}^{(n)}_{\ell}\left(\mathbf{x}_{\ell}\right)=\min_{\mathbf{P}\in\mathcal{P}}\mathbb{E}_\ell\left[{c}^{(n)}_{\ell}\left(\mathbf{P},\mathbf{x}_{\ell}\right)+{v}^{(n)}_{\ell+1}\left(\mathbf{x}_{\ell+1}\right)\right]:=\min_{\mathbf{P}\in\mathcal{P}}{V}^{(n)}_{\ell}\left(\mathbf{P},\mathbf{x}_{\ell}\right):=\min_{\mathbf{P}\in\mathcal{P}}\sum_{i\in\mathcal{S}}x_{i,\ell}V^{(n)}_{i,\ell}(\mathbf{P})$. Note that \eqref{global} and \eqref{local}  have different objective functions with the same set of constraints. The following theorem points out the equivalence between \eqref{global} and \eqref{local}.

\begin{theorem}\label{equivalent}
\textit{The optimal solutions to \eqref{global} collectively for $1\leq\ell<L$ and \eqref{local} coincide.}
\end{theorem}
\begin{proof}
First realize that \eqref{local} is equivalent to a myopic case. Therefore, the optimal value of \eqref{local} is given by
\begin{equation*}
	v^{*}_\ell(\mathbf{x})=\sum_{i\in\mathcal{S}}x_iq_{i,\ell}-\sum_{i\in\mathcal{S}}x_i\bar{\gamma}_{i,\ell}\ln\left(\sum_{j\in\mathcal{S}}e^{(\bar{\mu}_{ij,\ell}-\bar{v}_{j,\ell+1})/\bar{\gamma}_{i,\ell}}\right),
\end{equation*}
where $	\bar{v}_{i,\ell} = \frac{1}{N}\sum_{n\in\mathcal{N}}{v}^{(n)}_{i,\ell}$, and
\begin{equation}\label{valuefns}
	{v}^{(n)}_{i,\ell}=q_{i,\ell}-{\gamma}_{i,\ell}^{(n)}\ln\left(\sum_{j\in\mathcal{S}}e^{(\gamma_i^{(n)}\ln\bar{p}_{ij}^{(n)} -{v}^{(n)}_{j,\ell+1})/\gamma_i^{(n)}}\right).
\end{equation}
By induction and since $	{v}_{L}(\mathbf{x})=\frac{1}{N}\sum_{n\in\mathcal{N}}{v}^{(n)}_{L}(\mathbf{x})=\frac{1}{N}\sum_{n\in\mathcal{N}}c^{(n)}_L(\mathbf{x})$,
we have in stage $L-1$ that $v^{*}_{L-1}(\mathbf{x}) =v_{L-1}(\mathbf{x})
$. Suppose that at time $\ell+1$, $v^{*}_{\ell+1}(\mathbf{x})=v_{\ell+1}(\mathbf{x})$, which implies $\bar{v}_{i,\ell+1}=v_{i,\ell+1}.$ In stage $\ell$,
\begin{equation*}
	v^{*}_\ell(\mathbf{x})=\sum_{i\in\mathcal{S}}x_i\left(q_{i,\ell}-\bar{\gamma}_{i,\ell}\ln\left(\sum_{j\in\mathcal{S}}e^{(\bar{\mu}_{ij,\ell}-\bar{v}_{j,\ell+1})/\bar{\gamma}_{i,\ell}}\right)\right)=v_{\ell}(\mathbf{x}).
\end{equation*}
This completes the proof.
\end{proof}
Theorem~\ref{equivalent} shows that with local consensus, we can distribute the consensus across all stages and have the same result as global consensus. While global and local consensus schemes yield the same result with prior consensus, it may not be true for the case with posterior consensus. This is because that posterior belief requires units to know the state of other units in every stage, i.e., the joint state information in each stage, yet global consensus only allows units to acquire the joint stage information at most once, as shown in Fig.\ref{timeline}.

\section{Distributed Projected Gradient Algorithm}

To enable decentralized control, we make the following assumption of the communication between units:

\begin{assumption}\label{Poi}
Each unit chooses to share information with another unit according to an independent Poisson process with rate $r$, $r>0$. The $m$-th unit shares information with the $n$-th unit with probability $g^{(mn)}\in(0,1)$, $m,n\in\mathcal{N}$.\footnote{This assumption indicates that the vector $\mathbf{g}^{(m)}:=[g^{(mn)}]_{n\in\mathcal{N}}$ is a probability vector, i.e., $\mathbf{g}^{(m)}\in\Delta_{\mathcal{N}}$, $m\in\mathcal{N}$. It is worth noting that here we make a stronger assumption $g^{(mn)}\in(0,1)$ in contrast to that in convention $g^{(mn)}\in[0,1]$ or $g^{(mn)}$ is time variant for two reasons: 1) OpenADR 2.0 supports that each unit is able to communicate with all the other units, and 2) the relaxation of conditions on $g^{(mn)}$ needs more sophisticated technical treatment, which is beyond the scope of this paper. }
\end{assumption}
With this assumption, the arrival of the interactions between two arbitrary units is a Poisson process with rate $Nr$ guaranteeing that the probability of more than one unit being active at the same time is 0, thus avoiding data collisions.  
Assumption~\ref{Poi} also allows us to discretize the continuous time index by the number of arrivals. 
We let units exchange the policy matrices with each other and update their local policy matrices. We specify the exchanged information shared between units by discussing the difference in algorithm implementation between global and local consensus.

With the global consensus scheme, let $\mathbf{P}^{(n)}(t)=\big[\mathbf{P}^{(n)}_\ell(t)\big]_{1\leq \ell<L}$ be the collection of {policy matrices} maintained by the unit $n$ at time $t$, where $t$ is the continuous time index. Denote the time instant of the $k$-th information exchange by $t_k$. By letting $\mathbf{P}^{(n)}(k)=\mathbf{P}^{(n)}(t_k) $, we propose the following algorithm to update $\mathbf{P}^{(n)}(k)$:
\begin{equation}\label{updatemat}
	\begin{aligned}
		&\mathbf{Q}^{(n)}(k)=\sum_{m=1}^Ng^{(mn)}\mathbf{P}^{(n)}(k),\\
		&\mathbf{P}^{(n)}({k+1})=\hat{\mathcal{T}}\left(\mathbf{Q}^{(n)}(k)-\alpha_{k}\nabla_{\mathbf{P}^{(n)}(k)}\mathbf{V}^{(n)}(\mathbf{Q}^{(n)}(k))\right),
	\end{aligned}
\end{equation}
where 
$\mathbf{V}^{(n)}(\mathbf{P}):=[V_{i,\ell}^{(n)}(\mathbf{P}_{\ell})]_{i\in\mathcal{S},1\leq \ell<L}:=[\mathbf{V}_{\ell}^{(n)}(\mathbf{P}_{\ell})]_{1\leq \ell<L}$
with $\mathbf{P}=[\mathbf{P}_\ell]_{1\leq \ell<L}$.
Here, $\hat{\mathcal{T}}=[\mathcal{T}]_{1\leq \ell<L}$, and $\mathcal{T}:\mathbb{R}^{S\times S}\rightarrow\mathcal{P}$ projects a $S$-by-$S$ matrix onto the space of policy matrices $\mathcal{P}$. Moreover, $\alpha_k$ is a non-negative number. Algorithm \eqref{updatemat} resembles the \textit{distributed projected subgradient algorithm} over undirected graphs \cite{lobel2010distributed}. 
Convergence of \eqref{updatemat} to the optimal solution of \eqref{priorop} is guaranteed by:
\begin{proposition}[\cite{lobel2010distributed}]\label{asuuu}
\textit{	Let $\{\mathbf{P}^{(n)}(k)\}$ be generated by \eqref{updatemat} with $\alpha_k$ satisfying $\lim_{k\rightarrow\infty}\alpha_k=0$ and $\sum_{k=1}^\infty\alpha_k=\infty$. Then, $
	\lim_{k\rightarrow\infty}\big\|\text{vec}\left(\mathbf{P}^{(n)}(k)-\mathbf{P}^*\right)\big\|^2=0,\ \textcolor{black}{\forall}\,n\in\mathcal{N},
$ where $\mathbf{P}^*$ is an optimal solution to \eqref{priorop}.}
\end{proposition}

Theorem~\ref{equivalent} guarantees the effectiveness of the local consensus scheme to distribute the communication overhead among all units without performance losses. As a result, the communicated matrix size for each unit changes from one $S\times\, S\times\, (L-1)$ matrix to $L-1$ matrices sized $S\times S$. 
With local consensus, for $n\in\mathcal{N}$ and $1\leq \ell<L$, units update their maintained policy matrix $\mathbf{P}_\ell^{(n)}(k)$ in the $\ell$-\textit{th} stage as
\allowdisplaybreaks
\begin{equation}\label{updatemat2}
	\begin{aligned}
		&\mathbf{Q}_\ell^{(n)}(k)=\sum_{m=1}^Ng^{(mn)}(k)\mathbf{P}_\ell^{(n)}(k),\\
		&\mathbf{P}_\ell^{(n)}(k+1)=\mathcal{T}\left(\mathbf{Q}_\ell^{(n)}(k)-\alpha_k\nabla_{\mathbf{P}_\ell^{(n)}(k)}\mathbf{V}_\ell^{(n)}(\mathbf{Q}_\ell^{(n)}(k))\right),
	\end{aligned}
\end{equation}%
\allowdisplaybreaks[0]%
and convergence of \eqref{updatemat2} is an immediate result of Proposition \ref{asuuu}. Algorithm \eqref{updatemat2} not only reduces the data load for each communication, but also only ask units only needs to perform $\nabla_{\mathbf{P}_\ell^{(n)}(k)}$ instead of $\nabla_{\mathbf{P}^{(n)}(t_k)}$, which changes the dimension of complexity from $S\times S\times ({L-1})$ to $L-1$ repetitions of ${S}\times {S}$.

\section{Simulations}


\textcolor{black}{Based on the proposed model, a unit first computes its value functions $v_{i,\ell}^{(n)}$ defined in \eqref{valuefns} for all $i\in\mathcal{S}$ and $1\leq \ell< L$, which can be stored in a table. Then, with global/local consensus, units update their policy matrix according to \eqref{updatemat}/\eqref{updatemat2}. The convergence threshold regarding the error defined in Proposition \ref{asuuu} for convergence is set to be 0.01. When the algorithm converges, each unit draws a number uniformly distributed in $[0, 1]$. Let $\xi$ be the number drawn by unit $n$ at state $i$ and $p^*_{i}=[p^*_{ij}]_{j\in\mathcal{S}}\in\Delta_\mathcal{S}$ be its policy. If $\xi\in\left[\sum_{j=1}^sp^*_{ij},\sum_{j=s+1}^Sp^*_{ij}\right)$, unit $n$ changes its state to $s$.}

We carry out our case study using electricity consumption data from a building of New York University located in Manhattan, NY, from 2018 \cite{hassan2020hierarchical}. Using this data, we generate an ensemble of 100 synthetic buildings. The aggregated power consumption data for this ensemble of 100 buildings is displayed in Fig.~\ref{defaultmat}a and is used to construct the MP with 20 Markovian states for the period from mid-June to late-September, when all DR events in this region occurred in 2018. The resulting default transition probabilities $\bar{\mathbf{P}}$ are shown in Fig.~\ref{defaultmat}. To model each unit's default transition matrix, we perturb the matrix $\bar{\mathbf{P}}$ using uniform noise supported on $[0,0.05]$ for each unit, and then project the perturbed matrix onto $\mathcal{P}$. The discomfort cost is uniformly chosen from the set $[16, 24]$. We randomly generate default policies and discomfort costs for $N=100$ units. For simulating a realistic DR event, we isolate the period from 11:00 AM to 3:00 PM on July 2nd, 2018 when such an event happened in Manhattan, NY and NYU buildings were called upon to provide DR. We price electricity consumption during the DR event using the real-world tariff from Consolidated Edison \cite{condison}. Beside, the duration of each stage is set to be \unit[15]{minutes}, and there are $L=20$ stages.

\begin{figure}[t]
	\begin{center}
		\includegraphics[width=0.46\textwidth]{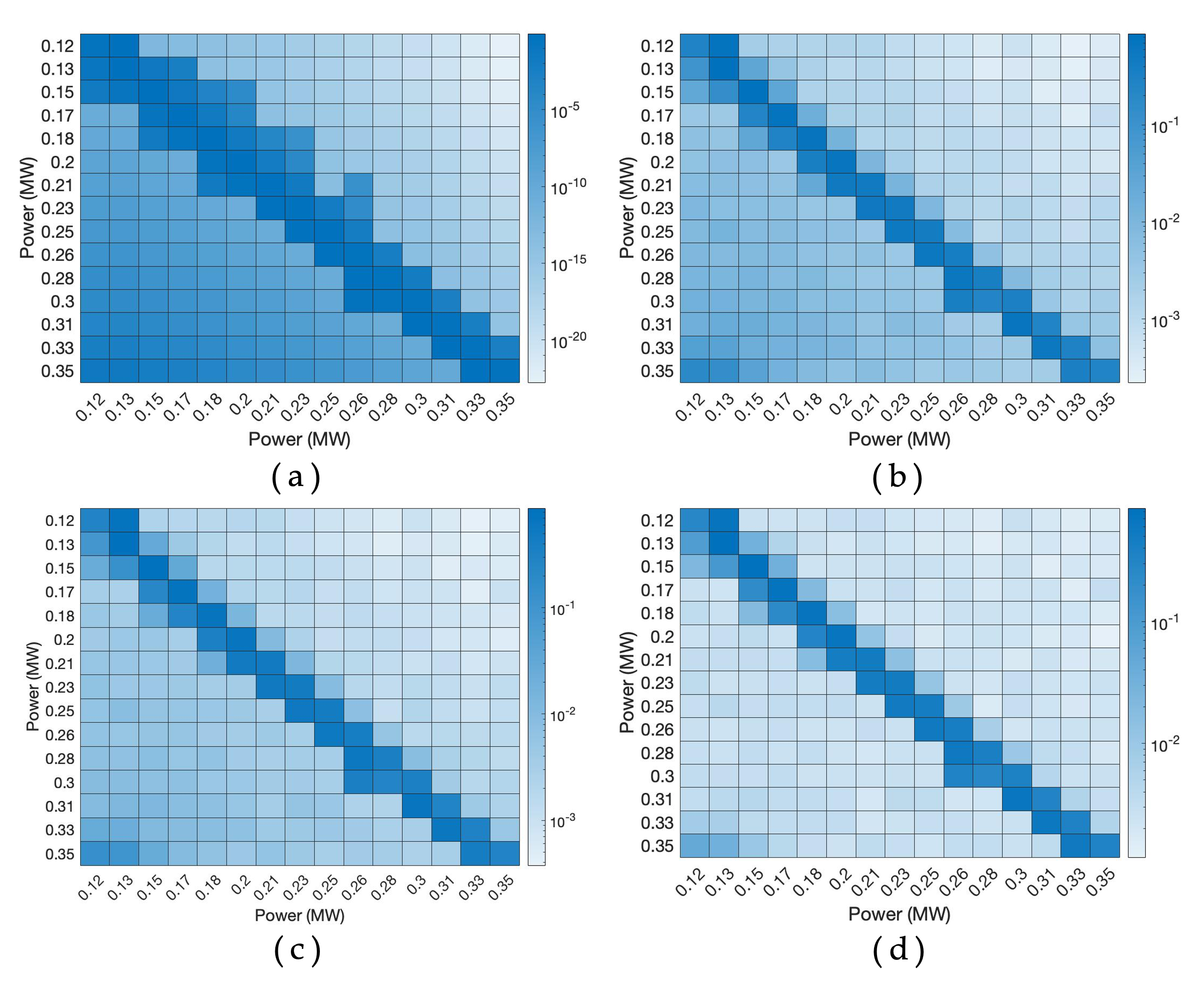}
		\caption{\textcolor{black}{Default transition probability matrix with 20 states constructed from the power profile versus optimally controlled policy matrices where color density indicates the probability value: (a) default matrix; optimal controlled matrix at (b) 12:00 p.m., (c) 1:00 p.m. and (d) 2:00 p.m.}}\label{defaultmat}
	\end{center}
\end{figure} 
Fig.~\ref{compare} shows the comparison of the average power consumption with different control schemes. The red line stands for the power consumption with the centralized control, while the purple and the black lines correspond to the decentralized control with local and global consensus schemes, respectively. As demonstrated in Fig.~\ref{compare}, the centralized and decentralized schemes yield similar consumption behaviors. Moreover, with the error \textcolor{black}{threshold} as 0.01, the global consensus and local consensus schemes yield similar performance.

\begin{figure}[t]
	\begin{center}
		\includegraphics[width=0.46\textwidth]{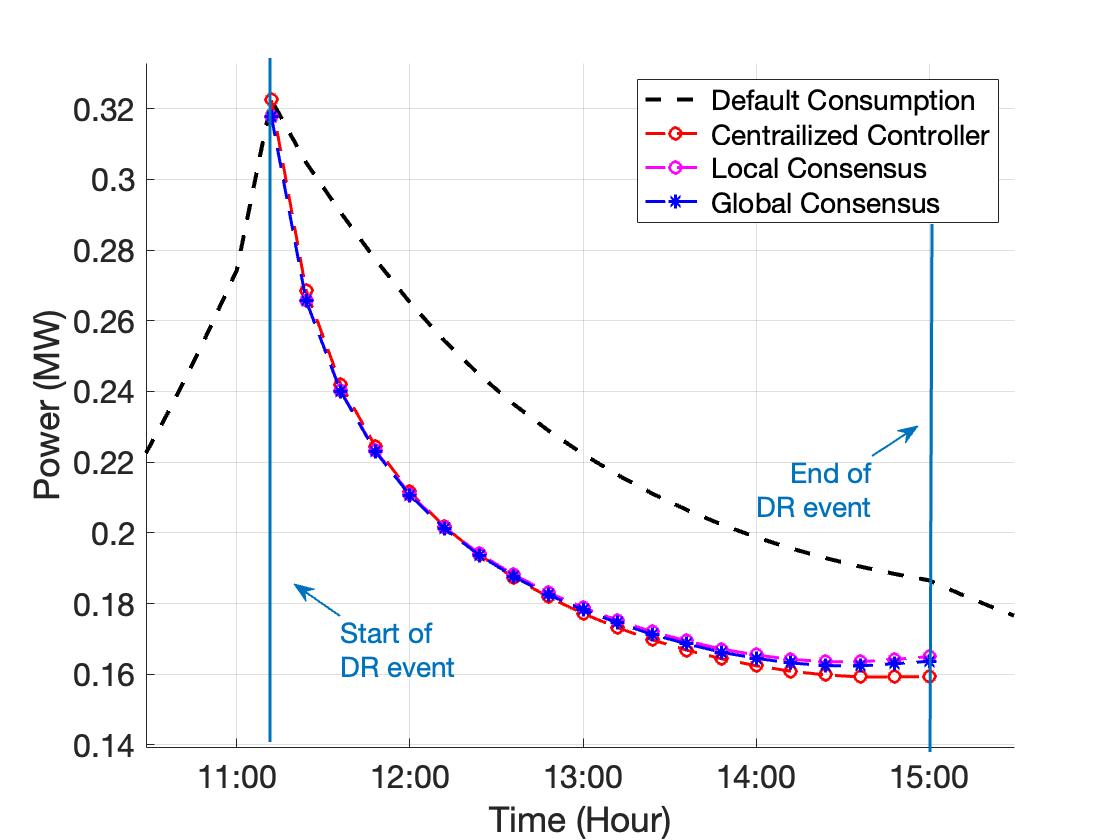}
		\caption{Comparison of average power consumption with different control schemes.}\label{compare}
	\end{center}
\end{figure} 


%
%
%

\section{Conclusions}

We proposed a decentralized consensus control for demand response events using Linearly Solvable Markov decision processes to capture the trade-off between energy consumption reduction and the  discomfort of DR participants.
Using the Bayesian approach, we have shown how prior consensus yields ensemble control policies that achieves consensus among all units. Further, we discussed global and local consensus schemes in the multistage setting and showed that units reach the same consensus under both schemes. 
Finally, we presented a gradient-based distributed control algorithm and demonstrated our results on a realistic case study.



\bibliographystyle{ieeetr} 
\bibliography{literature.bib} 

\end{document}